\documentclass[a4paper,UKenglish]{lipics}
 
\usepackage{microtype}

\title{Parameterized Approximation Algorithms for Packing Problems}
\titlerunning{Parameterized Approximation Algorithms for Packing Problems} 

\author[1]{Meirav Zehavi}
\affil[1]{Department of Computer Science, Technion IIT,
 Haifa 32000, Israel\\
 \texttt{meizeh@cs.technion.ac.il}}
\authorrunning{M. Zehavi} 

\Copyright{Meirav Zehavi}


\serieslogo{}
\volumeinfo
  {}
  {1}
  {}
  {1}
  {1}
  {1}
\EventShortName{}

\newcommand{\myparagraph}[1]{\par\smallskip\par\noindent{\bf{}#1:~}}

\usepackage{algorithm}
\usepackage{algorithmic}
\newcommand{\alg}[1]{\mbox{\sf #1}}  

\usepackage{amsopn}
\DeclareMathOperator{\family}{fam}

\begin{document}

\maketitle

\begin{abstract}
In the past decade, many parameterized algorithms were developed for packing problems.
Our goal is to obtain tradeoffs that improve the running times of these algorithms at the cost of computing approximate solutions.
Consider a packing problem for which there is no known algorithm with approximation ratio $\alpha$, and a parameter $k$. If the value of an optimal solution is at least $k$, we seek a solution of value at least $\alpha k$; otherwise, we seek an arbitrary solution. 
Clearly, if the best known parameterized algorithm that finds a solution of value $t$ runs in time $O^*(f(t))$ for some function $f$, we are interested in running times better than $O^*(f(\alpha k))$. We present tradeoffs between running times and approximation ratios for the {\sc $P_2$-Packing}, {\sc $3$-Set $k$-Packing} and {\sc $3$-Dimensional $k$-Matching} problems. Our tradeoffs are based on combinations of several known results, as well as a computation of ``approximate lopsided universal sets''.
\end{abstract}

\section{Introduction}

A problem is {\em fixed-parameter tractable (FPT)} with respect to a parameter $k$ if it can be solved in time $O^*(f(k))$ for some function $f$, where $O^*$ hides factors polynomial in the input size. Our goal is to improve the running times of parameterized algorithms for packing problems at the cost of computing approximate solutions. Consider a problem for which the best known polynomial-time approximation algorithm has approximation ratio $\beta$, as well as a parameter $k$. For any approximation ratio $\alpha$ that is better than $\beta$, if the value of an optimal solution is at least $k$, we seek a solution of value at least $\alpha k$, and otherwise we may return an arbitrary solution. Clearly, if the best known parameterized algorithm that finds a solution of value $t$ runs in time $O^*(f(t))$ for some function $f$, we are interested in running times better than $O^*(f(\alpha k))$.

We present tradeoffs between running times and approximation ratios in the context of the well-known {\sc $P_2$-Packing}, {\sc $3$-Set $k$-Packing} and {\sc $3$-Dimensional $k$-Matching ($3$D $k$-Matching)} problems, which are defined as follows.

\myparagraph{$P_2$-Packing} Given an undirected graph $G=(V,E)$ and a parameter $k\in\mathbb{N}$, we seek (in $G$) a set of $k$ (node-)disjoint simple paths on 3 nodes.

\myparagraph{$3$-Set $k$-Packing} Given a universe $E$, a family $\cal S$ of subsets of size 3 of $E$ and a parameter $k\in\mathbb{N}$, we seek a subfamily ${\cal S}'\subseteq {\cal S}$ of $k$ disjoint sets.

\myparagraph{$3$D $k$-Matching} Given disjoint universes $E_1$, $E_2$ and $E_3$, a family $\cal S$ of subsets of size 3 from $E_1\times E_2\times E_3$ and a parameter $k\in\mathbb{N}$, we seek a subfamily ${\cal S}'\subseteq {\cal S}$ of $k$ disjoint sets.

When we address the tradeoff versions of the above problems, we add $\alpha$ to their names. For example, given an instance $(G,k)$ of {\sc $P_2$-Packing}, as well as an accuracy parameter $\alpha\leq 1$, if $G$ has at least $k$ disjoint simple paths on 3 nodes, the {\sc $(\alpha,P_2)$-Packing} problem seeks a set of at least $\alpha k$ disjoint simple paths on 3 nodes, and otherwise it seeks an arbitrary set of such paths.

\subsection{Related Work}

The {\sc $3$-Set $k$-Packing}, {\sc $P_2$-Packing} and {\sc $3$D $k$-Matching} are well-studied problems, not only in the field of Parameterized Complexity. For example, the question of finding the largest $3$D-matching is a classic optimization problem, whose decision version is listed as one of the six fundamental NP-complete problems in Garey and Johnson~\cite{GareyJohnson79}. Clearly, {\sc $3$D Matching} is a special case of {\sc $3$-Set $k$-Packing}. By associating a set of three elements with every simple path on three nodes in a graph, it is also easy to see that {\sc $P_2$-Packing} is a special cases of {\sc $3$-Set $k$-Packing}.

In the past decade, the {\sc $3$-Set $k$-Packing} problem has enjoyed a race towards obtaining the fastest parameterized algorithm that solves it (see \cite{bjo10,impdetmatpac,chenalgorithmica2004,divandcol,fellowsalgorithmica2008,koutis2005,liutamc2007,multilineardetection,chenipec,wangcocoon2008,wangtamc2008,corrmatchpack,mixing}). Currently, the best deterministic algorithm runs in time $O^*(8.097^k)$ \cite{mixing}, and the best randomized algorithm runs in time $O^*(3.3432^k)$ \cite{bjo10}. Specialized parameterized algorithms for {\sc $P_2$-Packing} were given in \cite{p2packdet,p2packrand,p2packraible,p2pack2006,mixing}. Currently, the best deterministic algorithm runs in time $O^*(6.75^k)$ \cite{mixing} (based on \cite{p2packdet}), and the best randomized algorithm is the one for {\sc $3$-Set $k$-Packing} \cite{bjo10} (which runs in time $O^*(3.3432^k)$). Moreover, specialized parameterized algorithms for {\sc $3$D $k$-Matching} were given in \cite{bjo10,impdetmatpac,CLLSZ12,fsttcs13,KW09,chenipec,mixing}. Currently, the best deterministic algorithm runs in time $O^*(2.5961^{2k})$ \cite{mixing} (based on \cite{fsttcs13}), and the best randomized algorithm runs in time $O^*(2^k)$ \cite{bjo10}. Finally, we note that the best known (polynomial-time) approximation algorithm for {\sc $3$-Set $k$-Packing} has approximation ratio $\frac{3}{4}-\epsilon$ \cite{approxPackingbest}. This is also the best known (polynomial-time) approximation algorithm for {\sc $P_2$-Packing} and {\sc $3$D $k$-Matching}.
 
\subsection{Our Contribution and Organization}

In Section \ref{section:preliminaries}, we give necessary definitions and notation, including the definition of lopsided universal sets (of \cite{representative}). Then, in Section \ref{sec:approximateLopsided}, we define ``approximate lopsided universal sets'', and show how to compute them efficiently. In Section \ref{section:p2pack}, we develop a tradeoff-based algorithm for {\sc $P_2$-Packing}, which relies on two procedures: the main procedure combines a result by Feng {\em et al.}~\cite{p2packdet} with our computation of approximate universal sets; the second procedure (which also solves {\sc $3$-Set $(\alpha,k)$-Packing}) combines a partial execution of a known representative sets-based algorithm (from \cite{mixing}) and a known approximation algorithm by Cygan~\cite{approxPackingbest}. Section \ref{sec:Packing} presents a tradeoff-based algorithm for {\sc $3$-Set $k$-Packing}, which also relies on two procedures: the main procedure combines a simple and useful observation with algorithms from \cite{bjo10} and \cite{corrmatchpack}; the second procedure is the above mentioned second procedure of Section \ref{section:p2pack}. Finally, Appendix \ref{app:match} gives a tradeoff-based algorithm for {\sc $3$D $k$-Matching}, which is based on the same technique as the algorithm in Section \ref{sec:Packing}. The ideas underlying the design of our algorithms are intuitive and quite general, and may be used to develop parameterized approximation algorithms for other problems.

\section{Preliminaries}\label{section:preliminaries}

\myparagraph{Universal Sets} Roughly speaking, a lopsided universal set is a family of subsets, such that for any choice of disjoint sets $X$ and $Y$ of certain sizes, it contains a subset that captures all of the elements in $X$, but none of the elements in $Y$. Formally, it is defined as follows.

\begin{definition}\label{def:universalSet}
Given a universe $E$ of size $n$, we say that a family ${\cal F}\subseteq 2^E$ is an {\em $(n,k,p)$-universal set} if it satisfies the following condition: For every pair of sets $X\subseteq E$ of size $p$ and $Y\subseteq E\setminus X$ of size $k-p$, there is a set $F\in{\cal F}$ such that $X\subseteq F$ and $Y\cap F=\emptyset$.
\end{definition}

By the next result (of \cite{representative}), small lopsided universal sets can be computed efficiently.

\begin{theorem}[\cite{representative}]\label{theorem:universalSet}
There is a deterministic algorithm that computes an $(n,k,p)$-universal set of size $O(\displaystyle{{k\choose p}2^{o(k)}\log n})$ in time $O(\displaystyle{{k\choose p}2^{o(k)}n\log n})$.
\end{theorem}

\myparagraph{Representative Sets} A representative family (in the context of uniform matroids) is defined as follows.\footnote{We added (in Definition \ref{def:repfam}) the reference to the universe $E'$, which does not appear in the definition of a representative family of \cite{representative}, to simplify the presentation of the paper.}

\begin{definition}\label{def:repfam}
Given universes $E'\subseteq E$, a family ${\cal S}$ of subsets of size $p$ of $E$, and a parameter $k\in\mathbb{N}$, we say that a subfamily $\widehat{\cal S}\subseteq{\cal S}$ {\em $(k-p)$-represents} $\cal S$ with respect to $E'$ if for any pair of sets $X\in{\cal S}$ and $Y\subseteq E'\setminus X$ such that $|Y|\leq k-p$, there is a set $\widehat{X}\in\widehat{\cal S}$ disjoint from $Y$.
\end{definition}

Roughly speaking, this definition implies that if a set $Y$ can be extended to a set of size at most $k$ by adding a set $X\in{\cal S}$, then it can also be extended to a set of the same size by adding a set $\widehat{X}\in\widehat{\cal S}$. Many dynamic programming-based parameterized algorithms rely on computations of representative sets to speed-up their running times. We will use partial executions of such algorithms as black boxes.

\myparagraph{Notation} Given a graph $G=(V,E)$, a {\em $P_2$-Packing} is a set of disjoint paths (in $G$) on 3 nodes. Moreover, a 3-set is a set of 3 elements, and given a family $\cal S$ of 3-subsets, a 3-set packing is a subfamily of disjoint 3-sets from $\cal S$.

\section{Approximate Lopsided Universal Sets}\label{sec:approximateLopsided}

We first generalize Definition \ref{def:universalSet} to be suitable for approximation algorithms. The new definition makes use of an accuracy parameter, $0<\alpha\leq 1$. When $\alpha=1$, we obtain Definition \ref{def:universalSet}, and otherwise we obtain a more relaxed definition.

\begin{definition}\label{def:appUniversalSet}
Given a universe $E$ of size $n$, we say that a family ${\cal F}\subseteq 2^E$ is an {\em $(n,k,p,\alpha)$-universal set} if it satisfies the following condition: For every pair of sets $X\subseteq E$ of size $p$ and $Y\subseteq E\setminus X$ of size $k-p$, there is a set $F\in{\cal F}$ such that $|X\cap F|\geq\alpha p$, and $Y\cap F=\emptyset$.
\end{definition}

Now, we claim that small approximate lopsided universal sets (i.e., $(n,k,p,\alpha)$-universal sets) can be computed efficiently. Observe that when $\alpha=1$, we obtain the result stated in Theorem \ref{theorem:universalSet}.

\begin{theorem}\label{theorem:appUniversalSet}
There is a deterministic algorithm that computes an $(n,k,p,\alpha)$-universal set of size $O(\displaystyle{\frac{{k \choose \alpha p}}{{p\choose \alpha p}}2^{o(k)}\log n})$ in time $O(\displaystyle{\frac{{k \choose \alpha p}}{{p\choose \alpha p}}2^{o(k)}n\log n})$.
\end{theorem}

The proof of the above theorem is based on the proof of Theorem \ref{theorem:universalSet} (given in \cite{representative}). That is, we generalize the arguments given in \cite{representative}, taking into account the accuracy parameter $\alpha$. Towards the proof of Theorem \ref{theorem:appUniversalSet}, we need to prove three lemmas. Then, by repeatedly applying these lemmas, we will be able to prove the correctness of Theorem \ref{theorem:appUniversalSet}. We start with a lemma that presents an algorithm that is very slow, but computes approximate lopsided universal sets of the desired~size.

\begin{lemma}\label{lemma:set1}
There is a deterministic algorithm that computes an $(n,k,p,\alpha)$-universal set of size $\zeta(n,k,p,\alpha)$ in time $\tau(n,k,p,\alpha)$, where
\begin{itemize}
\item $\zeta(n,k,p,\alpha)=O(\displaystyle{\frac{{k \choose \alpha p}}{{p\choose \alpha p}}\cdot k^{O(1)}\log n})$.
\item $\tau(n,k,p,\alpha)=O(\displaystyle{{2^n \choose \zeta(n,k,p,\alpha)}\cdot n^{O(k)}})$.
\end{itemize}
\end{lemma}

\begin{proof}
First, we give a randomized algorithm which constructs, with positive probability, an $(n,k,p,\alpha)$-universal set of the desired size, $\zeta$. We then show how to deterministically construct an $(n,k,p,\alpha)$-universal set of the desired size, $\zeta$, in the desired time, $\tau$. Let $\displaystyle{t=\frac{\frac{k^k}{(\alpha p)^{\alpha p}(k-\alpha p)^{k-\alpha p}}}{{p\choose \alpha p}}(k+1)\ln n}$, and construct the family ${\cal F}=\{F_1,\ldots,F_t\}$ as follows. For each $i\in\{1,\ldots,t\}$ and element $e\in E$, insert $e$ to $F_i$ with probability $\displaystyle{\frac{\alpha p}{k}}$. The construction of different sets in ${\cal F}$, as well as the insertion of different elements into each set in ${\cal F}$, are independent. Clearly, $\zeta(n,k,p,\alpha)=t$ is within the required bound.

For fixed sets $X\subseteq E$ of size $p$, $Y\subseteq E\setminus X$ of size $k-p$, and $F\in{\cal F}$, the probability that $|X\cap F|=\alpha p$ and $Y\cap F=\emptyset$ is $\displaystyle{{p \choose \alpha p}\cdot(\frac{\alpha p}{k})^{\alpha p}(1-\frac{\alpha p}{k})^{k-\alpha p}}=\displaystyle{{p \choose \alpha p}\cdot\frac{(\alpha p)^{\alpha p}(k-\alpha p)^{k-\alpha p}}{k^k} =}$ $\displaystyle{\frac{(k+1)\ln n}{t}}$. Thus, the probability that {\em no} set $F\in{\cal F}$ satisfies $X\subseteq F$ and $Y\cap F=\emptyset$ is $\displaystyle{(1 - \frac{(k+1)\ln n}{t})^t\leq e^{-(k+1)\ln n}}=n^{-k-1}$. There are at most $n^k$ choices for $X$ and $Y$ as specified above; thus, applying the union bound, the probability that there exist such $X$ and $Y$ for which there no set $F\in{\cal F}$ that satisfies $|X\cap F|\geq\alpha p$ and $Y\cap F=\emptyset$, is at most $n^{-k-1}\cdot n^k=1/n$.

So far, we have given a randomized algorithm that constructs an $(n,k,p,\alpha)$-universal set of the desired size, $\zeta$, with probability at least $1-1/n>0$. To deterministically construct $\cal F$ in time bounded by $\tau$, we iterate over all families of $t$ subsets of $E$ (there are ${2^n\choose \zeta}$ such families), where for each family $\cal F$, we test in time $n^{O(k)}$ whether for any pair of sets $X\subseteq E$ of size $p$ and $Y\subseteq E\setminus X$ of size $k-p$, there is a set $F\in{\cal F}$ such that $|X\cap F|\geq\alpha p$ and $Y\cap F=\emptyset$.
\end{proof}

Next, we present a lemma using which we will be able to improve the running time of the algorithm in Lemma \ref{lemma:set1}. The proof of this lemma is almost identical to the proof of the corresponding lemma in \cite{representative}. For the sake of completeness, we give the proof in Appendix \ref{app:set2}.

\begin{lemma}\label{lemma:set2}
Given a deterministic algorithm that computes an $(n,k,p,\alpha)$-universal set of size $\zeta(n,k,p,\alpha)$ in time $\tau(n,k,p,\alpha)$, there is a deterministic algorithm that computes an $(n,k,p,\alpha)$-universal set of size $\zeta'(n,k,p,\alpha)$ in time $\tau'(n,k,p,\alpha)$, where
\begin{itemize}
\item $\zeta'(n,k,p,\alpha)=O(\zeta(k^2,k,p,\alpha)\cdot k^{O(1)}\log n)$.
\item $\tau'(n,k,p,\alpha)=O(\tau(k^2,k,p,\alpha)+\zeta'(n,k,p,\alpha)\cdot n)$.
\end{itemize}
\end{lemma}

Next, we present another lemma, which is also necessary to improve the running time of the algorithm in Lemma \ref{lemma:set1}.  Again, the proof of this lemma is almost identical to the proof of the corresponding lemma in \cite{representative}. For the sake of completeness, we give the proof in Appendix \ref{app:set3}. In this lemma, $s=\lfloor(\log k)^2\rfloor$ and $t=\lceil k/s\rceil$. Moreover, we let ${\cal Z}_{s,t}^p$ denote the set of all $t$-tuples $(p_1,p_2,\ldots,p_t)$ of integers such that $\sum_{i=1}^tp_i=p$, and $0\leq p_i\leq s$ for all $i\in\{1,2,\ldots,t\}$. Clearly, $|{\cal Z}_{s,t}^p|\leq \displaystyle{{p+t-1 \choose t-1}}\leq 2^{t\log(t+p)}$.

\begin{lemma}\label{lemma:set3}
Given a deterministic algorithm that computes an $(n,k,p,\alpha)$-universal set of size $\zeta(n,k,p,\alpha)$ in time $\tau(n,k,p,\alpha)$, there is a deterministic algorithm that computes an $(n,k,p,\alpha)$-universal set of size $\zeta'(n,k,p,\alpha)$ in time $\tau'(n,k,p,\alpha)$, where
\begin{itemize}
\item $\zeta'(n,k,p,\alpha)=O(\displaystyle{2^{O(t\log n)}\cdot\sum_{(p_1,\ldots,p_t)\in{\cal Z}_{s,t}^p}\prod_{i=1}^t\zeta(n,s,p_i,\alpha)})$.
\item $\tau'(n,k,p,\alpha)=O(\displaystyle{\sum_{\widehat{p}=1}^s\tau(n,s,\widehat{p},\alpha)+\zeta'(n,k,p,\alpha)\cdot n^{O(1)}})$.
\end{itemize}
\end{lemma}

We now turn to prove Theorem \ref{theorem:appUniversalSet}. Recall that the proof is structured as follows. We start by considering the algorithm in Lemma \ref{lemma:set1}, and then we repeatedly apply Lemmas \ref{lemma:set2} and \ref{lemma:set3} in order to obtain the desired algorithm.

\begin{proof}
First, by Lemma \ref{lemma:set1}, we have an algorithm that computes an $(n,k,p,\alpha)$-universal set of size $\zeta^1(n,k,p,\alpha)$ in time $\tau^1(n,k,p,\alpha)$, where
\begin{itemize}
\item $\zeta^1(n,k,p,\alpha)=O(\displaystyle{\frac{{k \choose \alpha p}}{{p\choose \alpha p}}\cdot k^{O(1)}\log n})$.
\item $\tau^1(n,k,p,\alpha)=O(\displaystyle{{2^n \choose \zeta^1(n,k,p,\alpha)}\cdot n^{O(k)}})$.
\end{itemize}

Observe that $\displaystyle{{2^{k^2} \choose \zeta^1(k^2,k,p,\alpha)}\cdot k^{O(k)} = 2^{k^{O(k)}}}$. Thus, by Lemma \ref{lemma:set2}, we have an algorithm that computes an $(n,k,p,\alpha)$-universal set of size $\zeta^2(n,k,p,\alpha)$ in time $\tau^2(n,k,p,\alpha)$, where
\begin{itemize}
\item $\zeta^2(n,k,p,\alpha)=O(\displaystyle{\frac{{k \choose \alpha p}}{{p\choose \alpha p}}\cdot k^{O(1)}\log n})$.
\item $\tau^2(n,k,p,\alpha)=O(\displaystyle{2^{k^{O(k)}} + \frac{{k \choose \alpha p}}{{p\choose \alpha p}}\cdot k^{O(1)}n\log n})$.
\end{itemize}

By applying Lemma \ref{lemma:set3}, we have an algorithm that computes an $(n,k,p,\alpha)$-universal set of size $\zeta^3(n,k,p,\alpha)$ in time $\tau^3(n,k,p,\alpha)$, where

\begin{itemize}
\item $\zeta^3(n,k,p,\alpha)=O(\displaystyle{2^{O(t\log n)}\cdot\sum_{(p_1,\ldots,p_t)\in{\cal Z}_{s,t}^p}\prod_{i=1}^t\zeta^2(n,s,p_i,\alpha)})$

$= O(\displaystyle{2^{O(t\log n)}\cdot\sum_{(p_1,\ldots,p_t)\in{\cal Z}_{s,t}^p}\prod_{i=1}^t\frac{{s \choose \alpha p_i}}{{p_i\choose \alpha p_i}}\cdot s^{O(1)}\log n})$

$= O(\displaystyle{2^{O(t\log n)}\cdot\max_{(p_1,\ldots,p_t)\in{\cal Z}_{s,t}^p}\prod_{i=1}^t\frac{{s \choose \alpha p_i}}{\frac{p_i^{p_i}}{(\alpha p_i)^{\alpha p_i}\cdot((1-\alpha)p_i)^{(1-\alpha)p_i}}}})$

$= O(\displaystyle{2^{O(t\log n)}\cdot\max_{(p_1,\ldots,p_t)\in{\cal Z}_{s,t}^p}\prod_{i=1}^t\frac{{s \choose \alpha p_i}}{\frac{1}{\alpha^{\alpha p_i}\cdot(1-\alpha)^{(1-\alpha)p_i}}}})$

$= O(\displaystyle{2^{O(t\log n)}\cdot(1-\alpha)^{(1-\alpha)p}\alpha^{\alpha p}\cdot\max_{(p_1,\ldots,p_t)\in{\cal Z}_{s,t}^p}\prod_{i=1}^t{s \choose \alpha p_i}})$

$= O(\displaystyle{2^{O(t\log n)}\cdot(1-\alpha)^{(1-\alpha)p}\alpha^{\alpha p}\cdot {k \choose \alpha p}})$.

\item $\tau^3(n,k,p,\alpha)=O(\displaystyle{\sum_{\widehat{p}=1}^s\tau^2(n,s,\widehat{p},\alpha)+\zeta^3(n,k,p,\alpha)\cdot n^{O(1)}})$

$=O(\displaystyle{2^{s^{O(s)}}+2^{O(t\log n)}\cdot(1-\alpha)^{(1-\alpha)p}\alpha^{\alpha p}\cdot {k \choose \alpha p}})$

$=O(\displaystyle{2^{(\log k)^{O(\log^2k)}}+2^{O(t\log n)}\cdot(1-\alpha)^{(1-\alpha)p}\alpha^{\alpha p}\cdot {k \choose \alpha p}})$.
\end{itemize}

Next, by applying Lemma \ref{lemma:set2} again, we have an algorithm that computes an $(n,k,p,\alpha)$-universal set of size $\zeta^4(n,k,p,\alpha)$ in time $\tau^4(n,k,p,\alpha)$, where
\begin{itemize}
\item $\zeta^4(n,k,p,\alpha)=O(\displaystyle{2^{O(\frac{k}{\log k})}\cdot(1-\alpha)^{(1-\alpha)p}\alpha^{\alpha p}\cdot {k \choose \alpha p} \cdot\log n})$.
\item $\tau^4(n,k,p,\alpha)=O(\displaystyle{2^{(\log k)^{O(\log^2k)}} + 2^{O(\frac{k}{\log k})}\cdot(1-\alpha)^{(1-\alpha)p}\alpha^{\alpha p}\cdot {k \choose \alpha p}\cdot n\log n})$.
\end{itemize}

Also, by applying Lemma \ref{lemma:set3} again, we have an algorithm that computes an $(n,k,p,\alpha)$-universal set of size $\zeta^5(n,k,p,\alpha)$ in time $\tau^5(n,k,p,\alpha)$, where
\begin{itemize}
\item $\zeta^5(n,k,p,\alpha)=O(\displaystyle{2^{O(t\log n)}\cdot\sum_{(p_1,\ldots,p_t)\in{\cal Z}_{s,t}^p}\prod_{i=1}^t\zeta^4(n,s,p_i,\alpha)})$

$=O(\displaystyle{2^{O(t\log n)}\cdot\sum_{(p_1,\ldots,p_t)\in{\cal Z}_{s,t}^p}\prod_{i=1}^t2^{O(\frac{s}{\log s})}\cdot(1-\alpha)^{(1-\alpha)p_i}\alpha^{\alpha p_i}\cdot {s \choose \alpha p_i} \cdot\log n})$

$=O(\displaystyle{2^{O(t\log n)}\cdot2^{O(\frac{k}{\log\log k})}\cdot(1-\alpha)^{(1-\alpha)p}\alpha^{\alpha p}\cdot\max_{(p_1,\ldots,p_t)\in{\cal Z}_{s,t}^p}\prod_{i=1}^t {s \choose \alpha p_i}})$

$=O(\displaystyle{2^{O(t\log n)}\cdot2^{O(\frac{k}{\log\log k})}\cdot(1-\alpha)^{(1-\alpha)p}\alpha^{\alpha p}\cdot{k \choose \alpha p}})$.

\item $\tau^5(n,k,p,\alpha)=O(\displaystyle{\sum_{\widehat{p}=1}^s\tau^4(n,s,\widehat{p},\alpha)+\zeta^4(n,k,p,\alpha)\cdot n^{O(1)}})$

$=O(\displaystyle{2^{(\log\log k)^{O(\log^2\log k)}} + 2^{O(t\log n)}\cdot2^{O(\frac{k}{\log\log k})}\cdot(1-\alpha)^{(1-\alpha)p}\alpha^{\alpha p}\cdot{k \choose \alpha p}})$

$=O(\displaystyle{2^{O(t\log n)}\cdot2^{O(\frac{k}{\log\log k})}\cdot(1-\alpha)^{(1-\alpha)p}\alpha^{\alpha p}\cdot{k \choose \alpha p}})$.
\end{itemize}

For the last transition above, observe that $\displaystyle{2^{(\log\log k)^{O(\log^2\log k)}}=2^{O(\frac{k}{\log\log k})}}$. Finally, by applying Lemma \ref{lemma:set2} again, we have an algorithm that computes an $(n,k,p,\alpha)$-universal set of size $\zeta^5(n,k,p,\alpha)$ in time $\tau^5(n,k,p,\alpha)$, where
\begin{itemize}
\item $\zeta^6(n,k,p,\alpha)=O(\displaystyle{2^{O(\frac{k}{\log\log k})}\cdot(1-\alpha)^{(1-\alpha)p}\alpha^{\alpha p}\cdot{k \choose \alpha p}\cdot\log n})$

$=O(\displaystyle{2^{O(\frac{k}{\log\log k})}\cdot\frac{{k \choose \alpha p}}{{p\choose\alpha p}}\cdot\log n})$.

\item $\tau^6(n,k,p,\alpha)=O(\displaystyle{2^{O(\frac{k}{\log\log k})}\cdot\frac{{k \choose \alpha p}}{{p\choose\alpha p}}\cdot n\log n})$.
\end{itemize}

The last algorithm is the desired one, which concludes the proof.
\end{proof}

\section{An Algorithm for {\sc $P_2$-Packing}}\label{section:p2pack}

In this section, we develop a parameterized algorithm that finds approximate solutions for {\sc $P_2$-Packing}. First, in Section \ref{section:packPro1}, we develop a procedure based on approximate lopsided universal sets and a polynomial-time algorithm by Feng {\em et al.}~\cite{p2packdet} for a special case of {\sc $P_2$-Packing}, which will be efficient when the value of $\alpha$ is large. For this procedure, \alg{Pack1}, we will prove the following result.

\begin{lemma}\label{lemma:packPro1}
Given an instance $(G=(V,E),k)$ of {\sc $P_2$-Packing}, as well as an accuracy parameter $\alpha\leq 1$, \alg{Pack1} solves {\sc $(\alpha,P_2)$-Packing} in deterministic time $O^*(2^{o(k)}\cdot\displaystyle{\frac{{3k \choose \alpha k}}{{k\choose \alpha k}}})$.
\end{lemma}

Second, in Section \ref{section:packPro2}, we develop a simple procedure based on an approximation algorithm for {\sc $3$-Set $k$-Packing}  by Cygan~\cite{approxPackingbest}, as well as a parameterized algorithm for this problem from \cite{corrmatchpack}, which will be efficient when the value of $\alpha$ is small. For this procedure, we will prove the following result.

\begin{lemma}\label{lemma:packPro2}
Given an instance $(E,{\cal S},k)$ of {\sc $3$-Set $k$-Packing}, as well as an accuracy parameter $0.75\leq\alpha\leq 1$, let $\beta^*=\frac{4\alpha-3+4\epsilon}{1+4\epsilon}$.\footnote{The parameter $\epsilon > 0$ can take any fixed value chosen by the user; for efficiency, the value should be small (close to $0$).} Then, given any $c\geq 1$, \alg{Pack2} solves {\sc $3$-Set $(\alpha,k)$-Packing} in deterministic time $O^*(2^{o(k)}\cdot\displaystyle{\max_{0\leq\beta\leq\beta^*}\left(\frac{(c(3-\beta))^{6-4\beta}}{(2\beta)^{2\beta}\cdot(c(3-\beta)-2\beta)^{6-6\beta}}\right)^k})$.
\end{lemma}

Since {\sc $P_2$-Packing} is a special case of {\sc $3$-Set $k$-Packing}, where one simply associates a 3-set with every simple path on three nodes, we obtain the following corollary.

\begin{corollary}\label{cor:packPro2}
Given an instance $(G,k)$ of {\sc $P_2$-Packing}, as well as an accuracy parameter $0.75\leq\alpha\leq 1$, let $\beta^*=\frac{4\alpha-3+4\epsilon}{1+4\epsilon}$. Then, given any $c\geq 1$, \alg{Pack2} solves {\sc $(\alpha,P_2)$-Packing} in deterministic time $O^*(2^{o(k)}\cdot\displaystyle{\max_{0\leq\beta\leq\beta^*}\left(\frac{(c(3-\beta))^{6-4\beta}}{(2\beta)^{2\beta}\cdot(c(3-\beta)-2\beta)^{6-6\beta}}\right)^k})$.
\end{corollary}

Recall that there is polynomial-time $(0.75-\epsilon)$-approximation algorithm for {\sc $P_2$-Packing} \cite{approxPackingbest}. Now, given a value $0.75\leq\alpha\leq 1$, we can simply call the procedure among \alg{Pack1} and \alg{Pack2} that is more efficient. Thus, we immediately obtain an algorithm, \alg{Pack}, for which we have the following result.

\begin{theorem}
Given an instance $(G,k)$ of {\sc $P_2$-Packing}, as well as an accuracy parameter $0.75\leq\alpha\leq 1$, let $\beta^*=\frac{4\alpha-3+4\epsilon}{1+4\epsilon}$. Then, given any $c\geq 1$, \alg{Pack} solves {\sc $(\alpha,P_2)$-Packing} in deterministic time $O^*(2^{o(k)}\cdot\min\left\{\displaystyle{\frac{{3k \choose \alpha k}}{{k\choose \alpha k}}},\displaystyle{\max_{0\leq\beta\leq\beta^*}\left(\frac{(c(3-\beta))^{6-4\beta}}{(2\beta)^{2\beta}\cdot(c(3-\beta)-2\beta)^{6-6\beta}}\right)^k}\right\})$.
\end{theorem}

Concrete figures for the running time of algorithm \alg{Pack} are given in Table~\ref{tab:packRunningTimes} (see Appendix \ref{app:tables}).

\subsection{The Procedure \alg{Pack1}}\label{section:packPro1}

To present our procedure, \alg{Pack1}, we need the following result by Feng {\em et al.}~\cite{p2packdet}, which solves a special case of {\sc $P_2$-Packing} in bipartite graphs in polynomial-time.

\begin{theorem}[\cite{p2packdet}]\label{theorem:packBip}
Given a bipartite graph $G=(L,R,E)$, there is a polynomial-time deterministic algorithm that finds a $P_2$-packing in $G$ of maximum size among all $P_2$-packings in $G$ that only contain paths whose middle vertices belong to $L$.
\end{theorem}

On a high-level, \alg{Pack1} uses an approximate lopsided universal set to create a set of inputs to the special case in Theorem \ref{theorem:packBip}, returning a large enough $P_2$-packing {\em iff} such a packing is a solution to one of the inputs. Now, we present the pseudocode of \alg{Pack1} (see Algorithm \ref{algorithm:pack1}), and give a more precise description. First, \alg{Pack1} obtains a $(|V|,3k,k,\alpha)$-universal set, $\cal F$ (Step 1). Then, it iterates over every set $F$ in $\cal F$ (Step 2). For each set $F$, it defines a bipartite graph $B$ by letting $L$ be $F$, $R$ be the set of the remaining vertices in $G$, and the set of edges contain every edge in $G$ that connects a node in $L$ with a node in $R$ (Step 3). It uses the algorithm in Theorem \ref{theorem:packBip} to compute a $P_2$-packing in $B$ (Step 4). If the packing contains enough paths (i.e., at least $\alpha k$ paths), \alg{Pack1} returns it (Step 5--6). Finally, if \alg{Pack1} did not find any large enough $P_2$-packing, it returns an empty one (Step 9).

\begin{algorithm}[!ht]
\caption{\alg{Pack1}($G=(V,E),k,\alpha$)}
\begin{algorithmic}[1]\label{algorithm:pack1}
\STATE Compute a $(|V|,3k,k,\alpha)$-universal set, $\cal F$, by using the algorithm in Theorem \ref{theorem:appUniversalSet}.
\FORALL{$F\in{\cal F}$}
	\STATE Define a bipartite graph $B=(F,V\setminus F,\{\{v,u\}\in E: v\in F, u\notin F\})$.
	\STATE Let $\cal P$ be a $P_2$-packing returned by the algorithm in Theorem \ref{theorem:packBip}, using the graph~$B$.
	\IF{$|{\cal P}|\geq\alpha k$}
		\STATE Return $\cal P$.
	\ENDIF
\ENDFOR
\STATE Return an empty $P_2$-packing.
\end{algorithmic}
\end{algorithm}

We now turn to prove the correctness of Lemma \ref{lemma:packPro1}.

\begin{proof}
First, to prove the correctness of \alg{Pack1}, we need to show that if $G$ has a $P_2$-packing of size at least $k$, then \alg{Pack1} returns a $P_2$-packing of size at least $\alpha k$. To this end, suppose that ${\cal P}^*$ is a $P_2$-packing of size $k$. Let $A$ denote the nodes that are middle nodes in the paths in ${\cal P}^*$, and let $B$ denote the other nodes in the paths in ${\cal P}^*$. Then, $|A|=k$ and $|B|=2k$. Therefore, since $\cal F$ is a $(|V|,3k,k,\alpha)$-universal set, there exists $F\in{\cal F}$ such that $|F\cap A|\geq\alpha k$ and $F\cap B=\emptyset$. Therefore, in the iteration the corresponds to $F$, we construct a bipartite graph $B=(L,R,E_B)$ such that at least $\alpha k$ paths in ${\cal P}^*$ have their middle nodes contained in $L$, and all the paths in ${\cal P}^*$, including those that have their middle nodes contained in $L$, have their endpoint nodes contained in $R$. Thus, by its correctness, the algorithm in Theorem \ref{theorem:packBip} returns a $P_2$-packing in $B$, which is also a $P_2$ packing in $G$ (since $B$ is a subgraph of $G$), of at least $\alpha k$ paths, which is then returned by \alg{Pack1}.

For the running time analysis, observe that by Theorem \ref{theorem:appUniversalSet}, \alg{Pack1} computes $\cal F$ (in Step 1) in time $O(\displaystyle{\frac{{3k \choose \alpha k}}{{k\choose \alpha k}}}\cdot 2^{o(k)}n\log n)$, and $|{\cal F}|=O(\displaystyle{\frac{{3k \choose \alpha k}}{{k\choose \alpha k}}}\cdot 2^{o(k)}\log n)$. Thus, since the algorithm in Theorem \ref{theorem:packBip} runs in polynomial-time, we conclude that \alg{Pack1} runs in the desired time.
\end{proof}

\subsection{The Procedure \alg{Pack2}}\label{section:packPro2}

To present our procedure, \alg{Pack2}, we need the following approximation algorithm by Cygan~\cite{approxPackingbest}.

\begin{theorem}[\cite{approxPackingbest}]\label{theorem:packApprox}
There is a deterministic polynomial-time approximation algorithm for {\sc $3$-Set Packing}, \alg{ApproxPack}, with approximation ratio $3/4-\epsilon$.
\end{theorem}

Assume an arbitrary order $<$ on $E$. Given a collection of families of sets, ${\bf S}$, let $\family({\bf S})=\{\bigcup{\cal S}: {\cal S}\in{\bf S}\}$ (i.e., we turn every family in ${\bf S}$ into a set). Moreover, given a family of sets, $\cal S$, let $\min({\cal S})=\{\min(S): S\in {\cal S}\}$ (i.e., we take each element that is the smallest element in some set in ${\cal S}$).
We also need the parameterized algorithm for {\sc $3$-Set $k$-Packing} of \cite{corrmatchpack}, for which we have the following result (augmented by the tradeoff-based computation of representative sets of \cite{productFam,repesa14}).

\begin{theorem}[\cite{corrmatchpack}, implicit]\label{theorem:packRepresent}
Let $(E,{\cal S},k)$ be an instance of {\sc $3$-Set $k$-Packing}, and let $0\leq \beta^*\leq 1$, $c\geq 1$ and $v\in E$. There is an algorithm, \alg{ParamPack}, which computes in time $T$ a collection of size at most $T$ of 3-set packings,\footnote{By \cite{corrmatchpack}, the size of $\widehat{\bf A}$ may be significantly smaller than $T$, but this will not be useful in our paper.} $\widehat{\bf A}\subseteq 2^{\cal S}$, such that $\family(\widehat{\bf A})$ $3(1-\beta^*)k$-represents $\cal A$ with respect to $\{u\in E: u > v\}$, where $T=O^*(\displaystyle{\max_{0\leq\beta\leq\beta^*}\left(\frac{(c(3-\beta))^{6-4\beta}}{(2\beta)^{2\beta}\cdot(c(3-\beta)-2\beta)^{6-6\beta}}\right)^k}\cdot 2^{o(k)})$ and ${\cal A}=\{\bigcup{\cal S}': {\cal S}'\subseteq{\cal S}, |{\cal S}'|=\beta^* k$, the sets in ${\cal S}'$ are disjoint$, \min({\cal S}')\subseteq\{u\in E: u\leq v\}\}$.
\end{theorem}

On a high-level, \alg{Pack2} calls \alg{ParamPack}, and attempts to complete the returned partial solutions by calling \alg{ApproxPack}. Now, we present the pseudocode of \alg{Pack2} (see Algorithm \ref{algorithm:pack2}), and give a more precise description. First, for all $v\in E$, \alg{Pack2} obtains a collection $\widehat{\bf A}_v$ such that $\family(\widehat{\bf A}_v)$ $3(1-\beta^*)k$-represents ${\cal A}$ (as defined in Theorem \ref{theorem:packRepresent}), where $\beta^* = \frac{4\alpha-3+4\epsilon}{1+4\epsilon}$ (Steps 1--4). It lets the collection $\widehat{\bf A}$ contain each family that belongs to a collection $\widehat{\bf A}_v$ for all $v\in E$ (Step 5). Then, it iterate over every family ${\cal P}'$ in $\widehat{\bf A}$ (Step 6). For each family ${\cal P}'$, it defines a family of 3-sets ${\cal B}$ that includes all the 3-sets in $\cal S$ that do not contain elements from $\bigcup{\cal P}'$ (Step 7). It uses the algorithm in Theorem \ref{theorem:packApprox} to compute a $3$-set packing, ${\cal P}$, in ${\cal B}$ (Step 8). If the combined packing, ${\cal P}'\cup{\cal P}$ contains enough 3-sets (i.e., at least $\alpha k$ 3-sets), \alg{Pack2} returns it (Steps 9--10). Finally, if \alg{Pack2} did not find any large enough 3-set packing, it returns an empty one (Step 13).

\begin{algorithm}[!ht]
\caption{\alg{Pack2}($E,{\cal S},k,\alpha$)}
\begin{algorithmic}[1]\label{algorithm:pack2}
\STATE Let $\beta^*\Leftarrow \frac{4\alpha-3+4\epsilon}{1+4\epsilon}$.
\FORALL{$v\in E$}
\STATE Compute a collection $\widehat{\bf A}_v$ such that $\family(\widehat{\bf A}_v)$ $3(1-\beta^*)k$-represents ${\cal A}$, which is defined in Theorem \ref{theorem:packRepresent}, by using the algorithm in this theorem.
\ENDFOR
\STATE Let $\widehat{\bf A}\Leftarrow \bigcup_{v\in E}\widehat{\bf A}_v$.
\FORALL{${\cal P}'\in\widehat{\bf A}$}
	\STATE Define ${\cal B}=\{S\in{\cal S}: S\cap (\bigcup{\cal P}')=\emptyset\}$.
	\STATE Let $\cal P$ be a 3-set packing returned by the algorithm in Theorem \ref{theorem:packApprox}, using the input $(E,{\cal B})$.
	\IF{$|{\cal P}'\cup{\cal P}|\geq\alpha k$}
		\STATE Return ${\cal P}'\cup{\cal P}$.
	\ENDIF
\ENDFOR
\STATE Return an empty 3-set packing.
\end{algorithmic}
\end{algorithm}

We now turn to prove the correctness of Lemma \ref{lemma:packPro2}.

\begin{proof}
Clearly, \alg{Pack2} returns only 3-set packings, since ${\cal P}'$ and $\cal P$ are 3-set packings (by Theorems  \ref{theorem:packApprox} and \ref{theorem:packRepresent}), and Step 7 ensures that ${\cal P}'\cup{\cal P}$ is also a 3-set packing. Thus, to prove the correctness of \alg{Pack2}, we need to show that if $\cal S$ has a 3-set packing of size at least $k$, then \alg{Pack2} returns a 3-set packing of size at least $\alpha k$. To this end, suppose that $\widetilde{\cal P}$ is a 3-set packing of size $k$. Observe that there exists $v\in E$, as well as a subset ${{\cal P}^*}'$ of $\beta^* k$ 3-sets from $\widetilde{\cal P}$, such that $\min({{\cal P}^*}')\subseteq\{u\in E: u\leq v\}$ and $\bigcup {\cal P}^*\subseteq \{u\in E: u > v\}$, where ${\cal P}^*=\widetilde{\cal P}\setminus {{\cal P}^*}'$. Then, $|\bigcup{\cal P}^*|=3(1-\beta^*)k$. Therefore, by Theorem \ref{theorem:packRepresent}, there exists ${\cal P}'$ in $\widehat{\bf A}_v\subseteq\widehat{\bf A}$ such that $(\bigcup{\cal P}')\cap(\bigcup{\cal P}^*)=\emptyset$. Consider the iteration to corresponds to ${\cal P}'$. Then, by Theorem \ref{theorem:packApprox}, \alg{Pack2} computes a 3-set packing $\cal P$ of size at least $(\frac{3}{4}-\epsilon)|\bigcup{\cal P}^*|=(\frac{3}{4}-\epsilon)(1-\beta^*)k$. Thus, \alg{Pack2} returns a 3-set packing of size $|{\cal P}'\cup{\cal P}|=|{\cal P}'|+|{\cal P}|\geq\beta^*k + (\frac{3}{4}-\epsilon)(1-\beta^*)k = [\frac{3}{4}-\epsilon+(\frac{1}{4}+\epsilon)\beta^*]k = [\frac{3}{4}-\epsilon+(\frac{1}{4}+\epsilon)\frac{4\alpha-3+4\epsilon}{1+4\epsilon}]k = \alpha k$.

For the running time analysis, observe that by Theorem \ref{theorem:appUniversalSet}, \alg{Pack2} computes $\widehat{\bf A}$ (in Steps 2--5) in time $T=O^*(\displaystyle{\left(\min_{1\leq c}\max_{0\leq\beta\leq\beta^*}\frac{(c(3-\beta))^{6-4\beta}}{(2\beta)^{2\beta}\cdot(c(3-\beta)-2\beta)^{6-6\beta}}\right)^k}\cdot 2^{o(k)})$, and $|\widehat{\bf A}|\leq T$. Thus, since the algorithm in Theorem \ref{theorem:packApprox} runs in polynomial-time, we conclude that \alg{Pack2} runs in the desired time.
\end{proof}

\section{An Algorithm for {\sc $3$-Set $k$-Packing}}\label{sec:Packing}

In this section, we develop a parameterized algorithm that finds approximate solutions for {\sc $3$-Set $k$-Packing}. We will develop two ``similar'' procedures, \alg{SetPack1} and \alg{SPRand1}, which will be efficient when the value of $\alpha$ is large. For these procedures, we will prove the following result.

\begin{lemma}\label{lemma:setPackPro1}
Given an instance $(E,{\cal S},k)$ of {\sc $3$-Set $k$-Packing}, as well as an accuracy parameter $0.75\leq\alpha\leq 1$, \alg{SetPack1} and \alg{SPRand1} solve {\sc $3$-Set $(\alpha,k)$-Packing} in deterministic time $O^*(8.097^{(1.5\alpha-0.5)k})$ and in randomized time $O^*(3.3432^{(1.5\alpha-0.5)k})$, respectively.
\end{lemma}

Recall that there is polynomial-time $(0.75-\epsilon)$-approximation algorithm for {\sc $3$-Set $k$-Packing} \cite{approxPackingbest}. Now, given a value $0.75\leq\alpha\leq 1$, we can simply call the procedure among \alg{SetPack1} (\alg{SPRand1}) and \alg{Pack2} (from Section \ref{section:p2pack}) that is more efficient. Thus, we immediately obtain algorithms, \alg{SetPack} and \alg{SPRand}, for which we have the following result.

\begin{theorem}
Given an instance $(E,{\cal S},k)$ of {\sc $3$-Set $k$-Packing}, and an accuracy parameter $0.75\leq\alpha\leq 1$, \alg{SetPack} and \alg{SPRand} solve {\sc $3$-Set $(\alpha,k)$-Packing} in deterministic time $O^*(\min\left\{\displaystyle{8.097^{(1.5\alpha-0.5)k}},2^{o(k)}\cdot\displaystyle{\max_{0\leq\beta\leq\beta^*}\left(\frac{(c(3-\beta))^{6-4\beta}}{(2\beta)^{2\beta}\cdot(c(3-\beta)-2\beta)^{6-6\beta}}\right)^k}\right\})$ and in randomized time $O^*(\min\left\{\displaystyle{3.3432^{(1.5\alpha-0.5)k}},2^{o(k)}\cdot\displaystyle{\max_{0\leq\beta\leq\beta^*}\left(\frac{(c(3-\beta))^{6-4\beta}}{(2\beta)^{2\beta}\cdot(c(3-\beta)-2\beta)^{6-6\beta}}\right)^k}\right\})$, respectively, for any $c\geq 1$, where $\beta^*=\frac{4\alpha-3+4\epsilon}{1+4\epsilon}$.
\end{theorem}

Concrete figures for the running time of algorithms \alg{SetPack} and \alg{SPRand} are given in Tables~\ref{tab:setPackRunningTimes} and \ref{tab:setPackRandRunningTimes} (see Appendix \ref{app:tables}), respectively.

We next turn to present \alg{SetPack1} and \alg{SPRand1}. To this end, we need the following results, given in \cite{mixing} and \cite{bjo10}.

\begin{theorem}[\cite{mixing}]\label{theorem:setPackParam}
There is a deterministic algorithm for {\sc $3$-Set $k$-Packing} that runs in time $O^*(8.097^k)$.
\end{theorem}

\begin{theorem}[\cite{bjo10}]\label{theorem:randSetPackParam}
There is a randomized algorithm for {\sc $3$-Set $k$-Packing} that runs in time $O^*(3.3432^k)$.
\end{theorem}

The pseudocode of \alg{SetPack1} is given below (see Algorithm \ref{algorithm:setPack1}). \alg{SPRand1} is identical to alg{SetPack1}, except that it calls the algorithm in Theorem \ref{theorem:randSetPackParam} rather than the algorithm in Theorem \ref{theorem:setPackParam}. On a high-level, \alg{SetPack1} creates an arbitrary small 3-set packing, and then attempts to complete it to a solution by calling the algorithm in Theorem \ref{theorem:setPackParam}. More precisely, \alg{SetPack1} first defines an empty 3-set packing ${\cal P}'$ (Step 1). Then, it iteratively attempts to add $\frac{(1-\alpha)k}{2}$ disjoint 3-sets from ${\cal S}$ to ${\cal P}'$ (Steps 2--8). To this end, at each iteration $i$, \alg{SetPack1} inserts (in Step 4) to ${\cal P}'$ an arbitrary 3-set $S$ from ${\cal S}$ that does not contain elements from any 3-set already in ${\cal P}'$. If such a set $S$ does not exist, \alg{SetPack1} simply returns an empty 3-set packing (Step 6). After \alg{SetPack1} finishes adding 3-sets to ${\cal P}'$, it lets $\widetilde{\cal S}$ contain the 3-sets in ${\cal S}$ that do not contain elements from any 3-set in ${\cal P}'$ (Step 9). Then, it attempts to find a 3-set packing ${\cal P}$ of size $(1.5\alpha-0.5)k$ in $\widetilde{\cal S}$ by calling the algorithm in Theorem \ref{theorem:setPackParam} (Step 10). Finally, it returns ${\cal P}'\cup{\cal P}$ (Step 11).

\begin{algorithm}[!ht]
\caption{\alg{SetPack1}($E,{\cal S},k,\alpha$)}
\begin{algorithmic}[1]\label{algorithm:setPack1}
\STATE ${\cal P}'\Leftarrow\emptyset$.
\FOR{$i=1,2,\ldots,\frac{(1-\alpha)k}{2}$}
	\IF{there exists $S\in{\cal S}$ such that $S\cap(\bigcup{\cal P}')=\emptyset$}
		\STATE Add $S$ to ${\cal P}'$.
	\ELSE
		\STATE Return an empty 3-set packing.
	\ENDIF
\ENDFOR
\STATE $\widetilde{\cal S}\Leftarrow \{S\in{\cal S}: S\cap(\bigcup{\cal P}')=\emptyset\}$.
\STATE Let ${\cal P}$ be a 3-set packing returned by the algorithm in Theorem \ref{theorem:setPackParam}, using the input $(E,\widetilde{\cal S},(1.5\alpha-0.5)k)$.
\STATE Return ${\cal P}'\cup{\cal P}$.
\end{algorithmic}
\end{algorithm}

We now prove the correctness of Lemma \ref{lemma:setPackPro1}.

\begin{proof}
Clearly, \alg{SetPack1} (\alg{SPRand1}) returns only 3-set packings, since by the pseudocode and Theorem \ref{theorem:setPackParam} (Theorem \ref{theorem:randSetPackParam}) ${\cal P}'$ and $\cal P$ are 3-set packings, and Step 9 ensures that ${\cal P}'\cup{\cal P}$ is also a 3-set packing. Thus, to prove the correctness of \alg{SetPack1} (\alg{SPRand1}), we need to show that if $\cal S$ has a 3-set packing of size at least $k$, then \alg{SetPack1} (\alg{SPRand1}) returns a 3-set packing of size at least $\alpha k$. To this end, suppose that $\widetilde{\cal P}$ is a 3-set packing of size $k$. Every 3-set in $\cal S$ can have a non-empty intersection with at most three 3-sets in $\widetilde{\cal P}$. Therefore, at each iteration $i$ (of Step 2), there exist at least $k-3i$ 3-sets in $\widetilde{\cal P}$ that do not contain elements that are contained in any 3-set in ${\cal P}'$. Thus, Step 6 is not executed. Moreover, after the last iteration of Step 2, $|{\cal P}'|=\frac{(1-\alpha)k}{2}$ and denoting ${{\cal P}^*}'=\{S\in\widetilde{\cal P}: S\cap(\bigcup{\cal P}')\neq\emptyset\}$, we have that $|{{\cal P}^*}'|\leq \frac{3(1-\alpha)k}{2}$. Denoting ${\cal P}^*=\widetilde{\cal P}\setminus{{\cal P}^*}'$, we have that $|{{\cal P}^*}|\geq k-\frac{3(1-\alpha)k}{2} = (1.5\alpha-0.5)k$. Observe that ${\cal P}^*\subseteq \widetilde{\cal S}$, where $\widetilde{\cal S}$ is defined in Step 9. Therefore, by Thereom \ref{theorem:setPackParam} (\ref{theorem:randSetPackParam}), \alg{SetPack1} (\alg{SPRand1}) obtains (in Step 10) a 3-set packing ${\cal P}$ of size $(1.5\alpha-0.5)k$. Thus, \alg{SetPack1} (\alg{SPRand1}) returns a 3-set packing of size $|{\cal P}'\cup{\cal P}| = |{\cal P}'|+|{\cal P}| = \frac{(1-\alpha)k}{2}+(1.5\alpha-0.5)k = \alpha k$.

For the running time analysis, observe that Steps 1--9 and 11 can be performed in deterministic polynomial-time. Moreover, by Theorem \ref{theorem:setPackParam} (\ref{theorem:randSetPackParam}), Step 10 can be performed in deterministic time $O^*(8.097^{(1.5\alpha-0.5)k})$ (randomized time $O^*(3.3432^{(1.5\alpha-0.5)k})$). Thus, \alg{SetPack1} (\alg{SPRand1}) runs in the desired time.
\end{proof}

\bibliography{References}

\newpage
\appendix

\section{Proof of Lemma \ref{lemma:set2}}\label{app:set2}

A family $\cal A$ of functions from $E$ to $\{1,2,\ldots,k^2\}$ is {\em $k$-perfect} if for every set $S\subseteq E$ of size $k$, there exists $f\in {\cal A}$ such that $f$ is injective when restricted to $S$. We start by obtaining such a family $\cal A$ of size $O(k^{O(1)}\log n)$ in time $O(k^{O(1)}n\log n)$ by using the construction by Alon {\em et al.}~\cite{colorcoding}.

For a set $S\subseteq E$ and a function $f\in{\cal A}$, define $f(S) = \{f(s) : s\in S\}$. Similarly, for a set $S\subseteq \{1,2,\ldots,k^2\}$, define $f^{-1}(S) = \{s\in S : f(s)\in S\}$. For a family $\cal S$ of subsets of $E$, define $f({\cal S})=\{f(S): S\in{\cal S}\}$. Similarly, for a family $\cal S$ of subsets of $\{1,2,\ldots,k^2\}$, define $f^{-1}({\cal S})=\{f^{-1}(S): S\in{\cal S}\}$.

Now, we use the given algorithm to contruct an $(k^2,k,p,\alpha)$-universal set, $\widehat{\cal F}$, of size $\zeta(k^2,k,p,\alpha)$ in time $\tau(k^2,k,p,\alpha)$ (with respect to the universe $\{1,2,\ldots,k^2\}$). Then, we let the desired $(n,k,p,\alpha)$-universal set be ${\cal F}=\displaystyle{\bigcup_{f\in{\cal A}}f^{-1}(\widehat{\cal F})}$.

Observe that $|{\cal F}|\leq |\widehat{\cal F}|\cdot k^{O(1)}\log n\leq O(\zeta(k^2,k,p,\alpha)\cdot k^{O(1)}\log n)$. Moreover, the computation of $\widehat{\cal F}$ is performed in time $O(\tau(k^2,k,p,\alpha)$, and then, the computation of $\cal F$ is performed in time $O(\zeta(k^2,k,p,\alpha)\cdot k^{O(1)}n\log n)$. Thus, we computed a family $\cal F$ of the desired size, $\zeta'$, in the desired time $\tau'$. It remains to show that $\widehat{\cal F}$ is an $(n,k,p,\alpha)$-universal set. Consider some sets $X\subseteq E$ of size $p$ and $Y\subseteq E\setminus X$ of size $k-p$. Since $\cal A$ is $k$-perfect, there is a function $f\in{\cal A}$ that is injective when restricted to $X\cup Y$. In particular, $f(X)\cap f(Y)=\emptyset$, $|f(X)|=p$ and $|f(Y)|=k-p$. Thus, since $\widehat{\cal F}$ is a $(k^2,k,p,\alpha)$-universal set, there exists $\widehat{F}\in\widehat{\cal F}$ such that $|\widehat{F}\cap f(X)|\geq \alpha p$ and $\widehat{F}\cap f(Y)=\emptyset$. Therefore, $|f^{-1}(\widehat{F})\cap X|\geq\alpha p$ and $f^{-1}(\widehat{F})\cap Y=\emptyset$. Since $f^{-1}(\widehat{F})\in{\cal F}$, we conclude that the lemma is correct.\qed

\section{Proof of Lemma \ref{lemma:set3}}\label{app:set3}

Let us denote $E=\{1,2,\ldots,n\}$. Correspondingly, let ${\cal P}_t$ denote the collection of all consecutive partitions of $E$ with exactly $t$ parts that are not necessarily non-empty. Clearly, $|{\cal P}_t|={n+t-1 \choose t-1}=2^{O(t\log n)}$. We will construct an $(n,st,p,\alpha)$-universal set, which is also an $(n,k,p,\alpha)$-universal set (since $st\geq k$).

For every $\widehat{p}\in\{0,1,\ldots,s\}$, we obtain an $(n,k,p,\alpha)$-universal set, $\widehat{\cal F}_{\widehat{p}}$, by using the given algorithm. Given a family ${\cal S}\subseteq 2^E$ and a set $S'\subseteq E$, define ${\cal S}\sqcap S' = \{S\cap S': S\in{\cal S}\}$. Moreover, given families ${\cal S},{\cal S}'\subseteq 2^E$, define ${\cal S}\circ{\cal S}' = \{S\cup S': S\in{\cal S}, S'\in{\cal S}'\}$. Now, we compute our $(n,st,p,\alpha)$-universal set, $\cal F$, by using the following formula.

\[\displaystyle{{\cal F} = \bigcup_{\begin{array}{l}
\{E_1,\ldots,E_t\}\in{\cal P}_t\\
(p_1,\ldots,p_t)\in{\cal Z}^p_{s,t}
\end{array}} (\widehat{\cal F}_{p_1}\sqcap E_1)\circ(\widehat{\cal F}_{p_2}\sqcap E_2)\circ\ldots\circ(\widehat{\cal F}_{p_t}\sqcap E_t)}.\]

By its definition, it immediately follows that $|{\cal F}|$ is within the desired bound. Moreover, the computation of the families $\widehat{\cal F}_{\widehat{p}}$ is done in time $O(\sum_{\widehat{p}=1}^s\tau(n,s,\widehat{p},\alpha))$. Afterwards, the computation of $\cal F$ is done in time $O(\zeta'(n,k,p,\alpha)\cdot n^{O(1)})$. Therefore, $\tau'(n,k,p,\alpha)$ is also within the desired bound. It remains to show that $\widehat{\cal F}$ is an $(n,st,p,\alpha)$-universal set. Consider some sets $X\subseteq E$ of size $p$ and $Y\subseteq E\setminus X$ of size $st-p$. There exists a consecutive partition $\{E_1,\ldots,E_t\}\in{\cal P}_t$ of $E$ such that for every $i\in\{1,\ldots,t\}$, we have that $|(X\cup Y)\cap E_i|=s$. For every $i\in\{1,\ldots,t\}$, let $p_i=|X\cap E_i|$. Since for every $i\in\{1,\ldots,t\}$, $\widehat{\cal F}_{p_i}$ is an $(n,s,p_i,\alpha)$-universal set, there exists $F_i\in\widehat{\cal F}_{p_i}$ such that $|F_i\cap (X\cap E_i)|\geq \alpha p_i$ and $F_i\cap (Y\cap E_i)=\emptyset$. Denote $F=(F_1\cap E_1)\cup(F_2\cap E_2)\cup\ldots\cup(F_t\cap E_t)$. Then, $|F\cap X|=\sum_{i=1}^t|F_i\cap (X\cap E_i)|\geq\sum_{i=1}^t\alpha p_i=\alpha p$, and $F\cap Y = \bigcup_{i=1}^t(F_i\cap(Y\cap E_i))=\emptyset$. Since $F\in{\cal F}$, we conclude that the lemma is correct.\qed

\section{An Algorithm for {\sc $3$-Dimensional $k$-Matching}}\label{app:match}

To obtain a parameterized algorithm that finds approximate solutions for {\sc $3$D $k$-Matching} (which is a special case of {\sc $3$-Set $k$-Packing}), we follow the arguments given in Sections \ref{section:packPro2} and \ref{sec:Packing}, replacing the best known algorithm for {\sc $3$-Set $k$-Packing} (that are used in these sections) by the best known algorithms for {\sc $3$D $k$-Matching}.

More precisely, in Section \ref{section:packPro2}, we now assume an arbitrary order $<$ on $E=E_1\cup E_2\cup E_3$ such that the elements in $E_1$ are the smallest (i.e., for all $v\in E_1$ and $u\in E_2\cup E_3$, we have that $v<u$). Instead of Theorem \ref{theorem:packRepresent}, we have the following result of \cite{fsttcs13} (augmented by the tradeoff-based computation of representative sets of \cite{productFam,repesa14}). 

\begin{theorem}[\cite{fsttcs13}, implicit]\label{theorem:matchRepresent}
Let $(E_1,E_2,E_3,{\cal S},k)$ be an instance of {\sc $3$D $k$-Matching}, and let $0\leq \beta^*\leq 1$, $c\geq 1$ and $v\in E_1$. There is an algorithm, \alg{ParamMatch}, which computes in time $T$ a collection of size at most $T$  of 3-set packings, $\widehat{\bf A}\subseteq 2^{\cal S}$, such that $\family(\widehat{\bf A})$ $2(1-\beta^*)k$-represents $\cal A$ with respect to $E_2\cup E_3$, where $T=O^*(\displaystyle{\max_{0\leq\beta\leq\beta^*}\left(\frac{c^{4-2\beta}}{\beta^{2\beta}\cdot(c-\beta)^{4-4\beta}}\right)^k}\cdot 2^{o(k)})$ and ${\cal A}=\{\bigcup{\cal S}': {\cal S}'\subseteq{\cal S}, |{\cal S}'|=\beta^* k$, the sets in ${\cal S}'$ are disjoint$, \min({\cal S}')\subseteq\{u\in E: u\leq v\}\}$.
\end{theorem}

Then, as shown in Section \ref{section:packPro2} (we need to use Theorem \ref{theorem:matchRepresent} rather than Theorem \ref{theorem:packRepresent}), we obtain a procedure \alg{Match2}, for which we have the following result.

\begin{lemma}\label{lemma:matchPro2}
Given an instance $(E_1,E_2,E_3,{\cal S},k)$ of {\sc $3$D $k$-Matching}, as well as an accuracy parameter $0.75\leq\alpha\leq 1$, let $\beta^*=\frac{4\alpha-3+4\epsilon}{1+4\epsilon}$. Then, for any $c\geq 1$, \alg{Pack2} solves {\sc $3$-Set $(\alpha,k)$-Packing} in deterministic time $O^*(2^{o(k)}\cdot\displaystyle{\max_{0\leq\beta\leq\beta^*}\left(\frac{c^{4-2\beta}}{\beta^{2\beta}\cdot(c-\beta)^{4-4\beta}}\right)^k})$.
\end{lemma}

In Section \ref{sec:Packing}, instead of Theorems \ref{theorem:setPackParam} and \ref{theorem:randSetPackParam}, we have the following theorems. 

\begin{theorem}[\cite{mixing}]\label{theorem:matchParam}
There is a deterministic algorithm for {\sc $3$D $k$-Matching} that runs in time $O^*(2.5961^{2k})$.
\end{theorem}

\begin{theorem}[\cite{bjo10}]\label{theorem:randMatchParam}
There is a randomized algorithm for {\sc $3$D $k$-Matching} that runs in time $O^*(2^k)$.
\end{theorem}

Then, as shown in Section \ref{sec:Packing} (we need to use Theorem \ref{theorem:matchParam} and Theorem \ref{theorem:randMatchParam} rather than Theorem \ref{theorem:setPackParam} and Theorem \ref{theorem:randSetPackParam}, respectively), we obtain procedures \alg{Match1} and \alg{MatchRand1}, for which we have the following result.

\begin{lemma}\label{lemma:matchPro1}
Given an instance $(E_1,E_2,E_3,{\cal S},k)$ of {\sc $3$D $k$-Matching}, as well as an accuracy parameter $0.75\leq\alpha\leq 1$, \alg{Match1} and \alg{MatchRand1} solve {\sc $3$D $(\alpha,k)$-Matching} in deterministic time $O^*(2.5961^{(3\alpha-1)k})$ and in randomized time $O^*(2^{(1.5\alpha-0.5)k})$, respectively
\end{lemma}

Recall that there is polynomial-time $(0.75-\epsilon)$-approximation algorithm for {\sc $3$D $k$-Matching} \cite{approxPackingbest}. Now, given a value $0.75\leq\alpha\leq 1$, we can simply call the procedure among \alg{Match1} (\alg{MatchRand1}) and \alg{Match2} that is more efficient. Thus, we immediately obtain algorithms, \alg{Match} and \alg{MatchRand}, for which we have the following result.

\begin{theorem}
Given an instance $(E_1,E_2,E_3,{\cal S},k)$ of {\sc $3$D $k$-Matching}, and an accuracy parameter $0.75<\alpha\leq 1$, \alg{Match} and \alg{MatchRand} solve {\sc $3$D $(\alpha,k)$-Matching} in deterministic time $O^*(\min\left\{\displaystyle{2.5961^{(3\alpha-1)k}},2^{o(k)}\cdot\displaystyle{\max_{0\leq\beta\leq\beta^*}\left(\frac{c^{4-2\beta}}{\beta^{2\beta}\cdot(c-\beta)^{4-4\beta}}\right)^k}\right\})$ and in randomized time $O^*(\min\left\{\displaystyle{2^{(1.5\alpha-0.5)k}},2^{o(k)}\cdot\displaystyle{\max_{0\leq\beta\leq\beta^*}\left(\frac{c^{4-2\beta}}{\beta^{2\beta}\cdot(c-\beta)^{4-4\beta}}\right)^k}\right\})$, respectively, for any $c\geq 1$, where $\beta^*=\frac{4\alpha-3+4\epsilon}{1+4\epsilon}$.
\end{theorem}

Concrete figures for the running time of algorithms \alg{Match} and \alg{MatchRand} are given in Tables~\ref{tab:matchRunningTimes} and \ref{tab:matchRandRunningTimes} (see Appendix \ref{app:tables}), respectively.

\section{Tables}\label{app:tables}

\begin{table}[center,h!]
\centering
\begin{tabular}{|c|c|c|c|c|}
	\hline
	$\alpha$ & \alg{Pack}     & \alg{Pack1}    & \alg{Pack2}; $c$    & $O^*(6.75^{\alpha k+o(k)})$ \\\hline\hline
	0.99     & $O^*(6.338^k)$ & $O^*(6.338^k)$ & $---$      & $O^*(6.623^k)$              \\\hline
	0.98     & $O^*(6.034^k)$ & $O^*(6.034^k)$ & $---$      & $O^*(6.498^k)$              \\\hline
	0.97     & $O^*(5.774^k)$ & $O^*(5.774^k)$ & $---$      & $O^*(6.375^k)$              \\\hline
	0.96     & $O^*(5.544^k)$ & $O^*(5.544^k)$ & $---$      & $O^*(6.254^k)$              \\\hline
	0.95     & $O^*(5.337^k)$ & $O^*(5.337^k)$ & $---$      & $O^*(6.136^k)$              \\\hline
	0.94     & $O^*(5.147^k)$ & $O^*(5.147^k)$ & $---$      & $O^*(6.020^k)$              \\\hline
	0.93     & $O^*(4.972^k)$ & $O^*(4.972^k)$ & $---$      & $O^*(5.906^k)$              \\\hline	
	0.92     & $O^*(4.809^k)$ & $O^*(4.809^k)$ & $---$      & $O^*(5.794^k)$              \\\hline			
	0.91     & $O^*(4.658^k)$ & $O^*(4.658^k)$ & $---$      & $O^*(5.685^k)$              \\\hline		
	0.9      & $O^*(4.516^k)$ & $O^*(4.516^k)$ & $---$      & $O^*(5.577^k)$              \\\hline
	0.89     & $O^*(4.383^k)$ & $O^*(4.383^k)$ & $---$      & $O^*(5.472^k)$              \\\hline	
	0.88     & $O^*(4.257^k)$ & $O^*(4.257^k)$ & $---$      & $O^*(5.368^k)$              \\\hline	
	0.87     & $O^*(4.138^k)$ & $O^*(4.138^k)$ & $---$      & $O^*(5.267^k)$              \\\hline	
	0.86     & $O^*(4.025^k)$ & $O^*(4.025^k)$ & $---$      & $O^*(5.167^k)$              \\\hline		
	0.85     & $O^*(3.918^k)$ & $O^*(3.918^k)$ & $---$      & $O^*(5.069^k)$              \\\hline	
	0.84     & $O^*(3.816^k)$ & $O^*(3.816^k)$ & $---$      & $O^*(4.972^k)$              \\\hline
	0.83     & $O^*(3.719^k)$ & $O^*(3.719^k)$ & $---$      & $O^*(4.879^k)$              \\\hline	
	0.82     & $O^*(3.627^k)$ & $O^*(3.627^k)$ & $O^*(5.692^k)$; 1.8 & $O^*(4.787^k)$              \\\hline	
  0.81     & $O^*(3.538^k)$ & $O^*(3.538^k)$ & $O^*(4.880^k)$; 1.8 & $O^*(4.697^k)$              \\\hline	
  0.8      & $O^*(3.454^k)$ & $O^*(3.454^k)$ & $O^*(4.098^k)$; 1.9 & $O^*(4.608^k)$              \\\hline		
	0.79     & $O^*(3.361^k)$ & $O^*(3.373^k)$ & $O^*(3.361^k)$; 1.9 & $O^*(4.521^k)$              \\\hline
	0.78     & $O^*(2.684^k)$ & $O^*(3.295^k)$ & $O^*(2.684^k)$; 1.9 & $O^*(4.435^k)$              \\\hline		
	0.77     & $O^*(2.073^k)$ & $O^*(3.220^k)$ & $O^*(2.073^k)$; 1.9 & $O^*(4.351^k)$              \\\hline
	0.76     & $O^*(1.527^k)$ & $O^*(3.149^k)$ & $O^*(1.527^k)$; 2.0 & $O^*(4.269^k)$              \\\hline			
\end{tabular}\bigskip
\caption{The running times of \alg{Pack}, \alg{Pack1}, \alg{Pack2} and the best {\em exact} deterministic algorithm for {\sc $P_2$-Packing} \cite{mixing} (based on \cite{p2packdet}), for different accuracy parameters $\alpha$. Entries marked with dashes are too large to be relevant to the running time of \alg{Pack}.}
\label{tab:packRunningTimes}
\end{table}

\begin{table}[center]
\centering
\begin{tabular}{|c|c|c|c|c|}
	\hline
	$\alpha$ & \alg{SetPack}  & \alg{SetPack1} & \alg{Pack2}; $c$    & $O^*(8.097^{\alpha k})$ \\\hline\hline
	0.99     & $O^*(7.847^k)$ & $O^*(7.847^k)$ & $---$      & $O^*(7.930^k)$              \\\hline
	0.98     & $O^*(7.605^k)$ & $O^*(7.605^k)$ & $---$      & $O^*(7.766^k)$              \\\hline
	0.97     & $O^*(7.370^k)$ & $O^*(7.370^k)$ & $---$      & $O^*(7.605^k)$              \\\hline
	0.96     & $O^*(7.143^k)$ & $O^*(7.143^k)$ & $---$      & $O^*(7.448^k)$              \\\hline
	0.95     & $O^*(6.922^k)$ & $O^*(6.922^k)$ & $---$      & $O^*(7.294^k)$              \\\hline
	0.94     & $O^*(6.708^k)$ & $O^*(6.708^k)$ & $---$      & $O^*(7.174^k)$              \\\hline
	0.93     & $O^*(6.501^k)$ & $O^*(6.501^k)$ & $---$      & $O^*(6.995^k)$              \\\hline	
	0.92     & $O^*(6.300^k)$ & $O^*(6.300^k)$ & $---$      & $O^*(6.850^k)$              \\\hline			
	0.91     & $O^*(6.106^k)$ & $O^*(6.106^k)$ & $---$      & $O^*(6.708^k)$              \\\hline		
	0.9      & $O^*(5.917^k)$ & $O^*(5.917^k)$ & $---$      & $O^*(6.569^k)$              \\\hline
	0.89     & $O^*(5.734^k)$ & $O^*(5.734^k)$ & $---$      & $O^*(6.433^k)$              \\\hline	
	0.88     & $O^*(5.557^k)$ & $O^*(5.557^k)$ & $---$      & $O^*(6.230^k)$              \\\hline	
	0.87     & $O^*(5.386^k)$ & $O^*(5.386^k)$ & $---$      & $O^*(6.170^k)$              \\\hline	
	0.86     & $O^*(5.219^k)$ & $O^*(5.219^k)$ & $---$      & $O^*(6.042^k)$              \\\hline		
	0.85     & $O^*(5.058^k)$ & $O^*(5.058^k)$ & $---$      & $O^*(5.917^k)$              \\\hline	
	0.84     & $O^*(4.902^k)$ & $O^*(4.902^k)$ & $---$      & $O^*(5.795^k)$              \\\hline
	0.83     & $O^*(4.751^k)$ & $O^*(4.751^k)$ & $---$      & $O^*(5.675^k)$              \\\hline	
	0.82     & $O^*(4.604^k)$ & $O^*(4.604^k)$ & $O^*(5.692^k)$; 1.8 & $O^*(5.557^k)$              \\\hline	
  0.81     & $O^*(4.462^k)$ & $O^*(4.462^k)$ & $O^*(4.880^k)$; 1.8 & $O^*(5.442^k)$              \\\hline	
  0.8      & $O^*(4.098^k)$ & $O^*(4.324^k)$ & $O^*(4.098^k)$; 1.9 & $O^*(5.330^k)$              \\\hline		
	0.79     & $O^*(3.361^k)$ & $O^*(4.190^k)$ & $O^*(3.361^k)$; 1.9 & $O^*(5.219^k)$              \\\hline
	0.78     & $O^*(2.684^k)$ & $O^*(4.061^k)$ & $O^*(2.684^k)$; 1.9 & $O^*(5.111^k)$              \\\hline		
	0.77     & $O^*(2.073^k)$ & $O^*(3.936^k)$ & $O^*(2.073^k)$; 1.9 & $O^*(5.006^k)$              \\\hline
	0.76     & $O^*(1.527^k)$ & $O^*(3.814^k)$ & $O^*(1.527^k)$; 2.0 & $O^*(4.902^k)$              \\\hline			
\end{tabular}\bigskip
\caption{The running times of \alg{SetPack}, \alg{SetPack1}, \alg{Pack2} and the best {\em exact} deterministic algorithm for {\sc $3$-Set $k$-Packing} \cite{mixing}, for different accuracy parameters $\alpha$. Entries marked with dashes are too large to be relevant to the running time of \alg{SetPack}.}
\label{tab:setPackRunningTimes}
\end{table}

\begin{table}[center]
\centering
\begin{tabular}{|c|c|c|c|c|}
	\hline
	$\alpha$ & \alg{SPRand}  & \alg{SPRand1} & \alg{Pack2}; $c$    & $O^*(3.3432^{\alpha k})$ \\\hline\hline
	0.99     & $O^*(3.2833^k)$ & $O^*(3.2833^k)$ & $---$      & $O^*(3.3031^k)$              \\\hline
	0.98     & $O^*(3.2244^k)$ & $O^*(3.2244^k)$ & $---$      & $O^*(3.2635^k)$              \\\hline
	0.97     & $O^*(3.1665^k)$ & $O^*(3.1665^k)$ & $---$      & $O^*(3.2244^k)$              \\\hline
	0.96     & $O^*(3.1097^k)$ & $O^*(3.1097^k)$ & $---$      & $O^*(3.1857^k)$              \\\hline
	0.95     & $O^*(3.0539^k)$ & $O^*(3.0539^k)$ & $---$      & $O^*(3.1475^k)$              \\\hline
	0.94     & $O^*(2.9991^k)$ & $O^*(2.9991^k)$ & $---$      & $O^*(3.1097^k)$              \\\hline
	0.93     & $O^*(2.9453^k)$ & $O^*(2.9453^k)$ & $---$      & $O^*(3.0724^k)$              \\\hline	
	0.92     & $O^*(2.8925^k)$ & $O^*(2.8925^k)$ & $---$      & $O^*(3.0355^k)$              \\\hline			
	0.91     & $O^*(2.8406^k)$ & $O^*(2.8406^k)$ & $---$      & $O^*(2.9991^k)$              \\\hline		
	0.9      & $O^*(2.7896^k)$ & $O^*(2.7896^k)$ & $---$      & $O^*(2.9631^k)$              \\\hline
	0.89     & $O^*(2.7396^k)$ & $O^*(2.7396^k)$ & $---$      & $O^*(2.9276^k)$              \\\hline	
	0.88     & $O^*(2.6904^k)$ & $O^*(2.6904^k)$ & $---$      & $O^*(2.8925^k)$              \\\hline	
	0.87     & $O^*(2.6422^k)$ & $O^*(2.6422^k)$ & $---$      & $O^*(2.8678^k)$              \\\hline	
	0.86     & $O^*(2.5948^k)$ & $O^*(2.5948^k)$ & $---$      & $O^*(2.8235^k)$              \\\hline		
	0.85     & $O^*(2.5482^k)$ & $O^*(2.5482^k)$ & $---$      & $O^*(2.7896^k)$              \\\hline	
	0.84     & $O^*(2.5025^k)$ & $O^*(2.5025^k)$ & $---$      & $O^*(2.7562^k)$              \\\hline
	0.83     & $O^*(2.4576^k)$ & $O^*(2.4576^k)$ & $---$      & $O^*(2.7231^k)$              \\\hline	
	0.82     & $O^*(2.4135^k)$ & $O^*(2.4135^k)$ & $O^*(5.6914^k)$; 1.8 & $O^*(2.6904^k)$              \\\hline	
  0.81     & $O^*(2.3702^k)$ & $O^*(2.3702^k)$ & $O^*(4.8798^k)$; 1.8 & $O^*(2.6582^k)$              \\\hline	
  0.8      & $O^*(2.3277^k)$ & $O^*(2.3277^k)$ & $O^*(4.0972^k)$; 1.9 & $O^*(2.6263^k)$              \\\hline		
	0.79     & $O^*(2.2859^k)$ & $O^*(2.2859^k)$ & $O^*(3.3607^k)$; 1.9 & $O^*(2.5948^k)$              \\\hline
	0.78     & $O^*(2.2449^k)$ & $O^*(2.2449^k)$ & $O^*(2.6838^k)$; 1.9 & $O^*(2.5636^k)$              \\\hline		
	0.77     & $O^*(2.0728^k)$ & $O^*(2.2046^k)$ & $O^*(2.0728^k)$; 1.9 & $O^*(2.5329^k)$              \\\hline
	0.76     & $O^*(1.5261^k)$ & $O^*(2.1651^k)$ & $O^*(1.5261^k)$; 2.0 & $O^*(2.5025^k)$              \\\hline					
\end{tabular}\bigskip
\caption{The running times of \alg{SPRand}, \alg{SPRand1}, \alg{Pack2} and the best {\em exact} randomized algorithm for {\sc $3$-Set $k$-Packing} \cite{bjo10}, for different accuracy parameters $\alpha$. Entries marked with dashes are too large to be relevant to the running time of \alg{SPRand}.}
\label{tab:setPackRandRunningTimes}
\end{table}

\begin{table}[center]
\centering
\begin{tabular}{|c|c|c|c|c|}
	\hline
	$\alpha$ & \alg{Match}  & \alg{Match1} & \alg{Match2}; $c$    & $O^*(2.5961^{2\alpha k})$ \\\hline\hline
	0.99     & $O^*(6.5496^k)$ & $O^*(6.5496^k)$ & $---$      & $O^*(6.6124^k)$              \\\hline
	0.98     & $O^*(6.3648^k)$ & $O^*(6.3648^k)$ & $---$      & $O^*(6.4874^k)$              \\\hline
	0.97     & $O^*(6.1853^k)$ & $O^*(6.1853^k)$ & $---$      & $O^*(6.3648^k)$              \\\hline
	0.96     & $O^*(6.0107^k)$ & $O^*(6.0107^k)$ & $---$      & $O^*(6.2445^k)$              \\\hline
	0.95     & $O^*(5.8411^k)$ & $O^*(5.8411^k)$ & $---$      & $O^*(6.1265^k)$              \\\hline
	0.94     & $O^*(5.6763^k)$ & $O^*(5.6763^k)$ & $---$      & $O^*(6.0107^k)$              \\\hline
	0.93     & $O^*(5.5162^k)$ & $O^*(5.5162^k)$ & $---$      & $O^*(5.8971^k)$              \\\hline	
	0.92     & $O^*(5.3606^k)$ & $O^*(5.3606^k)$ & $---$      & $O^*(5.7857^k)$              \\\hline			
	0.91     & $O^*(5.2093^k)$ & $O^*(5.2093^k)$ & $---$      & $O^*(5.6763^k)$              \\\hline		
	0.9      & $O^*(5.0623^k)$ & $O^*(5.0623^k)$ & $---$      & $O^*(5.5691^k)$              \\\hline
	0.89     & $O^*(4.9195^k)$ & $O^*(4.9195^k)$ & $---$      & $O^*(5.4638^k)$              \\\hline	
	0.88     & $O^*(4.7807^k)$ & $O^*(4.7807^k)$ & $---$      & $O^*(5.3606^k)$              \\\hline	
	0.87     & $O^*(4.6458^k)$ & $O^*(4.6458^k)$ & $---$      & $O^*(5.2592^k)$              \\\hline	
	0.86     & $O^*(4.5147^k)$ & $O^*(4.5147^k)$ & $---$      & $O^*(5.1598^k)$              \\\hline		
	0.85     & $O^*(4.3874^k)$ & $O^*(4.3874^k)$ & $---$      & $O^*(5.0623^k)$              \\\hline	
	0.84     & $O^*(4.2636^k)$ & $O^*(4.2636^k)$ & $---$      & $O^*(4.9667^k)$              \\\hline
	0.83     & $O^*(4.1433^k)$ & $O^*(4.1433^k)$ & $---$      & $O^*(4.8728^k)$              \\\hline	
	0.82     & $O^*(4.0264^k)$ & $O^*(4.0264^k)$ & $O^*(4.6105^k)$; 1.7 & $O^*(4.7807^k)$              \\\hline	
  0.81     & $O^*(3.9128^k)$ & $O^*(3.9128^k)$ & $O^*(4.0641^k)$; 1.8 & $O^*(4.6904^k)$              \\\hline	
  0.8      & $O^*(3.5107^k)$ & $O^*(3.8024^k)$ & $O^*(3.5107^k)$; 1.8 & $O^*(4.6017^k)$              \\\hline		
	0.79     & $O^*(2.9663^k)$ & $O^*(3.6951^k)$ & $O^*(2.9663^k)$; 1.8 & $O^*(4.5147^k)$              \\\hline
	0.78     & $O^*(2.4414^k)$ & $O^*(3.5908^k)$ & $O^*(2.4414^k)$; 1.9 & $O^*(4.4294^k)$              \\\hline		
	0.77     & $O^*(1.9448^k)$ & $O^*(3.4895^k)$ & $O^*(1.9448^k)$; 1.9 & $O^*(4.3457^k)$              \\\hline
	0.76     & $O^*(1.4778^k)$ & $O^*(3.3911^k)$ & $O^*(1.4778^k)$; 2.0 & $O^*(4.2636^k)$              \\\hline				
\end{tabular}\bigskip
\caption{The running times of \alg{Match}, \alg{Match1}, \alg{Match2} and the best {\em exact} deterministic algorithm for {\sc $3$D $k$-Matching} \cite{mixing}, for different accuracy parameters $\alpha$. Entries marked with dashes are too large to be relevant to the running time of \alg{Match}.}
\label{tab:matchRunningTimes}
\end{table}

\begin{table}[center]
\centering
\begin{tabular}{|c|c|c|c|c|}
	\hline
	$\alpha$ & \alg{MatchRand}  & \alg{MatchRand1} & \alg{Match2}; $c$    & $O^*(2^{\alpha k})$ \\\hline\hline
	0.99     & $O^*(1.9794^k)$ & $O^*(1.9794^k)$ & $---$      & $O^*(1.9862^k)$              \\\hline
	0.98     & $O^*(1.9589^k)$ & $O^*(1.9589^k)$ & $---$      & $O^*(1.9725^k)$              \\\hline
	0.97     & $O^*(1.9386^k)$ & $O^*(1.9386^k)$ & $---$      & $O^*(1.9589^k)$              \\\hline
	0.96     & $O^*(1.9186^k)$ & $O^*(1.9186^k)$ & $---$      & $O^*(1.9454^k)$              \\\hline
	0.95     & $O^*(1.8987^k)$ & $O^*(1.8987^k)$ & $---$      & $O^*(1.9319^k)$              \\\hline
	0.94     & $O^*(1.8791^k)$ & $O^*(1.8791^k)$ & $---$      & $O^*(1.9186^k)$              \\\hline
	0.93     & $O^*(1.8597^k)$ & $O^*(1.8597^k)$ & $---$      & $O^*(1.9053^k)$              \\\hline	
	0.92     & $O^*(1.8404^k)$ & $O^*(1.8404^k)$ & $---$      & $O^*(1.8922^k)$              \\\hline			
	0.91     & $O^*(1.8214^k)$ & $O^*(1.8214^k)$ & $---$      & $O^*(1.8791^k)$              \\\hline		
	0.9      & $O^*(1.8026^k)$ & $O^*(1.8026^k)$ & $---$      & $O^*(1.8661^k)$              \\\hline
	0.89     & $O^*(1.7839^k)$ & $O^*(1.7839^k)$ & $---$      & $O^*(1.8532^k)$              \\\hline	
	0.88     & $O^*(1.7655^k)$ & $O^*(1.7655^k)$ & $---$      & $O^*(1.8404^k)$              \\\hline	
	0.87     & $O^*(1.7472^k)$ & $O^*(1.7472^k)$ & $---$      & $O^*(1.8277^k)$              \\\hline	
	0.86     & $O^*(1.7291^k)$ & $O^*(1.7291^k)$ & $---$      & $O^*(1.8151^k)$              \\\hline		
	0.85     & $O^*(1.7112^k)$ & $O^*(1.7112^k)$ & $---$      & $O^*(1.8026^k)$              \\\hline	
	0.84     & $O^*(1.6935^k)$ & $O^*(1.6935^k)$ & $---$      & $O^*(1.7901^k)$              \\\hline
	0.83     & $O^*(1.6760^k)$ & $O^*(1.6760^k)$ & $---$      & $O^*(1.7777^k)$              \\\hline	
	0.82     & $O^*(1.6587^k)$ & $O^*(1.6587^k)$ & $O^*(4.6105^k)$; 1.7 & $O^*(1.7655^k)$              \\\hline	
  0.81     & $O^*(1.6415^k)$ & $O^*(1.6415^k)$ & $O^*(4.0641^k)$; 1.8 & $O^*(1.7533^k)$              \\\hline	
  0.8      & $O^*(1.6246^k)$ & $O^*(1.6246^k)$ & $O^*(3.5107^k)$; 1.8 & $O^*(1.7412^k)$              \\\hline		
	0.79     & $O^*(1.6078^k)$ & $O^*(1.6078^k)$ & $O^*(2.9663^k)$; 1.8 & $O^*(1.7291^k)$              \\\hline
	0.78     & $O^*(1.5911^k)$ & $O^*(1.5911^k)$ & $O^*(2.4414^k)$; 1.9 & $O^*(1.7172^k)$              \\\hline		
	0.77     & $O^*(1.5747^k)$ & $O^*(1.5747^k)$ & $O^*(1.9448^k)$; 1.9 & $O^*(1.7053^k)$              \\\hline
	0.76     & $O^*(1.4778^k)$ & $O^*(1.5584^k)$ & $O^*(1.4778^k)$; 2.0 & $O^*(1.6935^k)$              \\\hline					
\end{tabular}\bigskip
\caption{The running times of \alg{MatchRand}, \alg{MatchRand1}, \alg{Match2} and the best {\em exact} randomized algorithm for {\sc $3$D $k$-Matching} \cite{bjo10}, for different accuracy parameters $\alpha$. Entries marked with dashes are too large to be relevant to the running time of \alg{MatchRand}.}
\label{tab:matchRandRunningTimes}
\end{table}

\end{document}